\documentclass[11pt]{article}

\usepackage{amsmath,amssymb,amsthm}
\usepackage{graphicx}
\usepackage{hyperref}
\usepackage{geometry}
\newtheorem{theorem}{Theorem}

\theoremstyle{thmstyleone}%

\newtheorem{Corollary}[theorem]{Corollary}%
\newtheorem{Lemma}[theorem]{Lemma}%

\theoremstyle{thmstylethree}%
\newtheorem{definition}{Definition}%

\begin{document}

\title{Error Bound Analysis of Physics-Informed Neural Networks-Driven $T_{2}$ Quantification in Cardiac Magnetic Resonance Imaging}

\author{%
Mengxue Zhang,
Qingrui Cai,
Yinyin Chen,
Hang Jin, Jianjun Zhou,\\
Qiu Guo, Peijun Zhao, Zhiping Mao, Xingxing Zhang, Yuyu Xia, Xianwang Jiang,\\
Qin Xu, Chunyan Xiong, Yirong Zhou, Chengyan Wang,
Xiaobo Qu
$^{*}$\\
[0.5ex]
\small
\\
\small
\texttt
{quxiaobo@xmu.edu.cn}
\thanks{%
This work was supported in part by the National Natural Science Foundation of China
(62331021, 62371410, 62122064), Natural Science Foundation of Fujian Province of China
(2023J02005), National Key R
\&
D Program of China (2023YFF0714200),
Zhou Yongtang Fund for High Talents Team (0621-Z0332004),
President Fund of Xiamen University (20720220063),
Industry-University Cooperation Projects of the Ministry of Education of China
(231107173160805), and the Xiamen University Nanqiang Outstanding Talents Program.
Mengxue Zhang, Qingrui Cai, and Yinyin Chen contributed equally to this work.
Corresponding author: Xiaobo Qu.\\
Mengxue Zhang, Qingrui Cai, Chunyan Xiong, Yirong Zhou and Xiaobo Qu are with School of Electronic Science and Engineering (National Model Microelectronics College), Xiamen University-Neusoft Medical Magnetic Resonance Imaging Joint Research and Development Center, Fujian Provincial Key Laboratory of Plasma and Magnetic Resonance, Xiamen University, Xiamen 361102, China (e-mail:quxiaobo@xmu.edu.cn). \\
Yinyin Chen and Hang Jin are with Department of Radiology, Zhongshan Hospital (Xiamen), Fudan University, and Department of Medical Imaging, Shanghai Medical School, Fudan University, Xiamen 361015, China. \\
Jianjun Zhou is with Department of Radiology, Zhongshan Hospital (Xiamen), Fudan University, Fujian Province Key Clinical Specialty for Medical Imaging and Xiamen Key Laboratory of Clinical Transformation of Imaging Big Data and Artificial Intelligence, Xiamen 361015, China.\\ 
Qiu Guo is with Department of Radiology, Xiang’an Hospital of Xiamen University, Xiamen 361102, China.\\ 
Peijun Zhao is with Radiology Department, the First Affiliated Hospital of Xiamen University, Xiamen 361003, China. \\
Zhiping Mao is with School of Mathematical Sciences, Eastern Institute of Technology,
Ningbo 315200, China. \\
Xingxing Zhang, Yuyu Xia, Xianwang Jiang and Qin Xu are with Xiamen University-Neusoft Medical Magnetic Resonance Imaging Joint Research and Development Center, Shanghai Neusoft Medical Technology Co. Ltd, Shanghai 200241, China. \\
Chengyan Wang is with Human Phenome Institute, Fudan University, Shanghai 200241, China. }
}

\date{}
\maketitle

\begin{abstract}
Physics-Informed Neural Networks (PINN) are emerging as a promising approach for quantitative parameter estimation of Magnetic Resonance Imaging (MRI). While existing deep learning methods can provide an accurate quantitative estimation of the $T_{2}$ parameter, they still require large amounts of training data and lack theoretical support and a recognized gold standard. Thus, given the absence of PINN-based approaches for $T_{2}$ estimation, we propose embedding the fundamental physics of MRI, the Bloch equation, in the loss of PINN, which is solely based on target scan data and does not require a pre-defined training database. Furthermore, by deriving rigorous upper bounds for both the $T_{2}$ estimation error and the generalization error of the Bloch equation solution, we establish a theoretical foundation for evaluating the PINN's quantitative accuracy. Even without access to the ground truth or a gold standard, this theory enables us to estimate the error with respect to the real quantitative parameter $T_{2}$. The accuracy of $T_{2}$ mapping and the validity of the theoretical analysis are demonstrated on a numerical cardiac model and a water phantom, where our method exhibits excellent quantitative precision in the myocardial $T_{2}$ range. Clinical applicability is confirmed in 94 acute myocardial infarction (AMI) patients, achieving low-error quantitative $T_{2}$ estimation under the theoretical error bound, highlighting the robustness and potential of PINN.
\end{abstract}


\section{Introduction}
\label{sec:introduction}
Quantitative magnetic resonance imaging (qMRI) can measure parameters that reflect the intrinsic characteristics of tissues \cite{zhang2022}. These quantitative parameters include $T_{1}$, $T_{2}$, $T_{2}^{*}$, and the extracellular volume fraction (ECV), which can be used to detect diffuse fibrosis \cite{yaman2020}, evaluate the degree of myocardial edema \cite{baessler2019}, quantify the tissue iron content \cite{cadour2022}, and reflect the degree of myocardial fibrosis \cite{cadour2022}, respectively. In this work, we focus on the $T_{2}$ (Fig. \ref{figure1}). It enables the non-invasive detection of myocardial edema, inflammation, and subtle tissue changes that are not apparent on conventional images, thus supporting the diagnosis of myocarditis, acute myocardial infarction, and transplant rejection \cite{obrien2022}, and may help to save the myocardium as early as possible \cite{tschope2021}.

To estimate $T_{2}$ value from magnetic resonance imaging (MRI) images, there are two common steps: First, derive a $T_{2}$ parametric signal model following the physical rule, the Bloch equation \cite{hin1983}, of MRI; second, fit the parametric signal model to obtain $T_{2}$ value \cite{fatemi2020}. The least square is the most common fitting approach \cite{poon1992}.

In the era of artificial intelligence, deep learning has been applied to improve the robustness and accuracy of $T_{2}$ estimation \cite{yang2023,cai2018,jun2024,choi2025,guo2023,qiu2024,guo2022,morales2024}. Depending on the choice of referenced $T_{2}$ values as training labels, current deep learning can be divided into 2 types: 
1) Synthetic data training \cite {yang2023, cai2018,guo2022,guo2023,morales2024}; 2) Label-free training \cite {jun2024,qiu2024,choi2025}. The former synthesizes millions of 1D $T_{2}$-weighted signals \cite{guo2023} or thousands of 2D $T_{2}$-weighted maps as training labels \cite {cai2018,yang2023}, achieving close or comparable $T_{2}$ values than the least square approach. The latter directly fits the $T_{2}$ value to the target scan-specific image or k-space through neural network, achieving better patient classification accuracy than that obtained from least square \cite{qiu2024} or faster quantification under accelerated imaging \cite{jun2024, choi2025}. Not limited to $T_{2}$, both types have the advantage of compatibility to multi-parametric quantification, e.g. $T_{1}$ and proton density \cite {yang2023, guo2023, jun2024, choi2025, qiu2024}. 

Label-free (unsupervised) methods have not been applied to cardiac $T_{2}$ mapping, as the imaging sequences, motion characteristics, and $T_{2}$ value ranges in the heart differ substantially from those in the brain, making direct comparison with these brain-based methods infeasible. Among cardiac-specific approaches, the Deepfittingnet (DFN) framework by Guo et al.\cite{guo2023} represents the most advanced cardiac-specific model trained on a large database of synthetic myocardial $T_{2}$ signal evolutions. Although synthetic data training —including DFN\cite{guo2022,guo2023,morales2024}— has been applied to up to 32 healthy subjects in one cohort, they depend on the accuracy and coverage of the synthetic training signals and the reliability on patients is still questioned \cite {yang2023}. Besides, the lack of an established clinical gold standard in qMRI makes it difficult to validate the quantitative accuracy of deep learning models. This issue is further compounded by the limited theoretical research on these models.

To avoid using the predefined training database, physics-informed neural networks (PINN) \cite{raissi2019}, \cite{karniadakis2021} explore known physical equations as prior information and embed them in the loss function of the network. PINN has shown great potential in qMRI beyond traditional deep learning methods. Van Herten et al. integrated physical models of contrast agent kinetics into PINN for myocardial perfusion quantification \cite{vanherten2023}. PINN has also been applied to estimate diffusion coefficients from MRI by modeling molecular transport in the human brain, demonstrating their utility in solving inverse problems in biomedical imaging \cite{zapf2022}. 

Theoretical analysis of PINN in qMRI has not been explored yet. PINN is essentially used to solve differential equations, particularly partial differential equations (PDE) and ordinary differential equations. As a type of deep learning model, the solution obtained by PINN would lead to errors. Recent theoretical studies have analyzed the generalization error of PINN solutions for this forward problem \cite{han2023posteriori}, \cite{mishra2023generalization_pdes}, showing that error estimates derived from discrete training points can be generalized to continuous domains. In addition, generalization bounds have also been derived for a class of inverse problems that miss boundary conditions \cite{mishra2022generalization_inverse}. However, theoretical analysis of PINN for inverse problems with unknown parameters remains lacking, both in terms of solution's generalization error and parameter estimation error.

Here, we propose a PINN-based method for $T_{2}$ mapping and derive the error bound of parameter estimation. The fundamental physics of MRI, Bloch equation, is embedded into a physics-informed loss function. By learning the Bloch equation through the network, quantitative $T_{2}$ values are estimated by directly solving the inverse problem of equations with the sampled data, rather than with a predefined database in other deep learning methods. For generalization error and parameter estimation error, our theoretical framework does not require the ground truth $T_{2}$ value or clinical gold standards, perfectly fitting for the characterization of cardiac tissue where ground truth is not available. The framework of this work and the structure of the PINN are summarized in Fig. \ref{figure1}(a) and (b).
 
A preliminary version of this study was posted as a preprint \cite{cai2023bloch}, which introduced the basic idea of applying PINN to $T_{2}$ quantification. The present manuscript provides a substantially expanded investigation, including new theoretical error bounds, a comparison with a DFN-based database method, and comprehensive quantitative, in vivo, and clinical validations.






\begin{figure}[]
  \centering
  \includegraphics[width=1\textwidth]{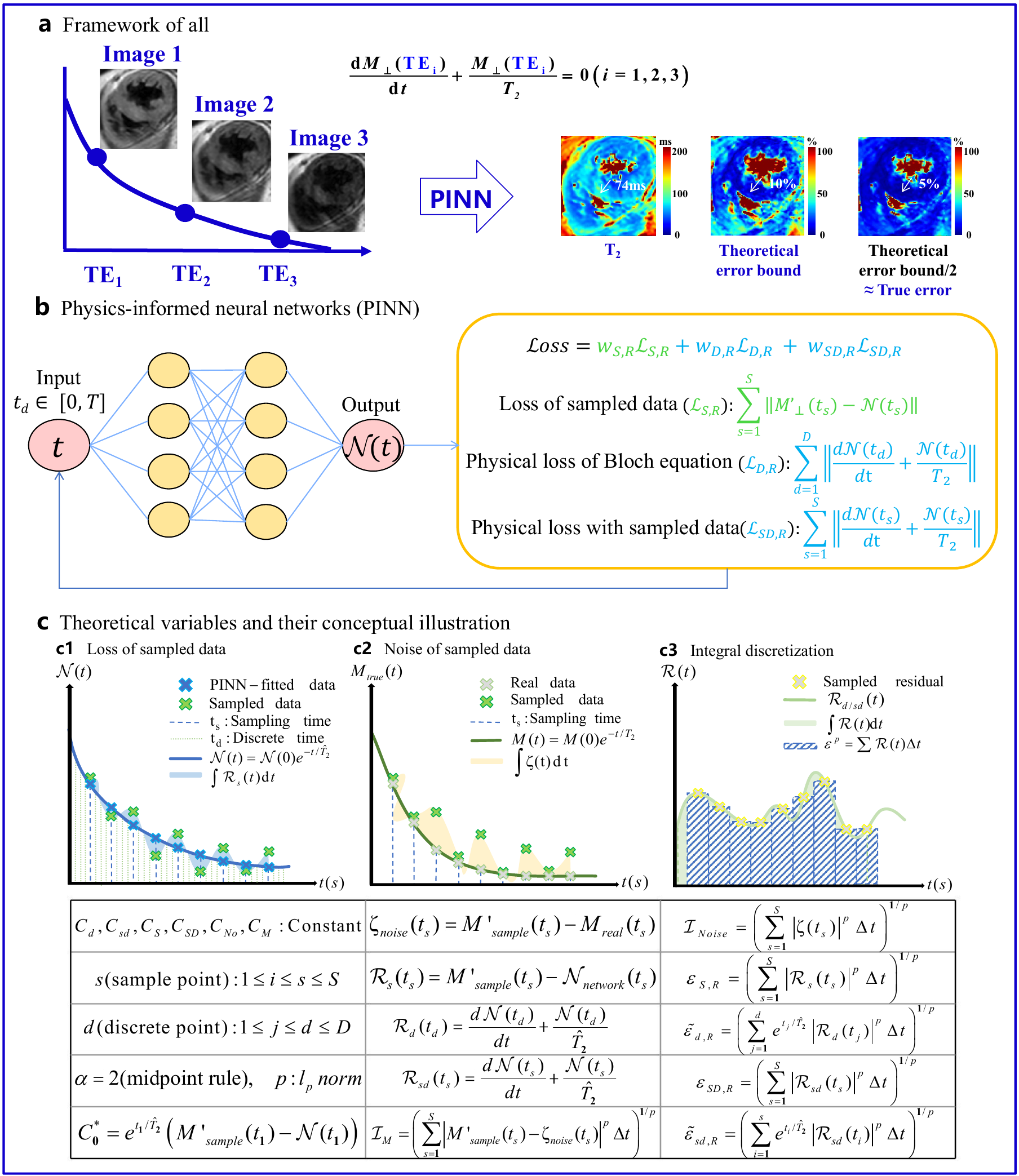}
  \caption{ Framework of the proposed cardiac $T_2$ mapping with physics-informed neural networks (PINN). (a) The framework and theoretical error bound analysis. Left: grayscale images at three echo times (TEs), where each pixel follows the Bloch equation solution $M_{\perp}(t)$ $ (t = TE_{1}, TE_{2}, TE_{3}) $ with exponential decay. PINN estimates pixel-wise $T_2$ by fitting $M_{\perp}(t)$. Middle: $T_2$ mapping estimated by PINN. Right: theoretical error bound of parameter estimation, and empirical finding shows that error bound/2 closely approximates the real error based on statistical analysis of numerical experiments in which the ground truth is available. (b) Network structure of the Bloch equation-based PINN. The input to the network consists of $d$ discrete time points uniformly sampled from the continuous domain $[0,~T]$, and the output $\mathcal{N}(t)$ is used to fit the real $M_{\perp}(t)$. The total loss function is composed of sampled data loss and physical loss. The third loss is the physics loss at the sampling time points, which is only extracted for use in the subsequent theoretical analysis. (c) Theoretical variables and their conceptual illustration. Note: The sampled data means the acquired data at each sampling time (TE). Discrete time is the time point for magnetization evolution.}
  \label{figure1}
\end{figure}

\section{Proposed Method}
In this section, we present our network structure and the underlying principles of the PINN method, along with the definition, lemma and derivation required for theoretical analysis.

\subsection{Bloch equation}
The magnetization vector  $\mathbf{M}=\left(M_{x}, M_{y}, M_{z}\right)^{T}$  in the magnetic field satisfies the Bloch equation \cite{hin1983}:
\begin{align}\label{1}
\frac{d \mathbf{M}(\mathbf{r}, t)}{d t} = & \gamma \mathbf{M}(\mathbf{r}, t) \times \mathbf{B}(\mathbf{r}, t)-\frac{M_{x}(\mathbf{r}, t) \mathbf{i}+M_{y}(\mathbf{r}, t) \mathbf{j}}{T_{2}}\\\nonumber
&-\frac{M_{z}(\mathbf{r}, t)-M_{0}}{T_{1}} \mathbf{k},
\end{align}
where $t\in[0,~T]$, $ \mathbf{B}(\mathbf{r}, t)$  is magnetic field at a spatial location  $\mathbf{r}$ while $M_{0}  $, $T_{1} $ and $T_{2} $ are tissue parameters. Due to the limitation of coil placement, in practical applications, the magnetization intensity  $M_{z}$  in the $z$-direction is difficult to be directly detected. Therefore, our target is the transverse magnetization vector  $M_{\perp}=\sqrt{M_{x}^{2}+M_{y}^{2}} $. For the spin echo and gradient echo sequence that measure the  $T_{2}$  value of tissues \cite{baessler2015}, the transverse magnetization vector  $M_{\perp}$  satisfies the Bloch equation that can be simplified as:
\begin{equation}\label{2}
\frac{d M_{\perp}(\mathbf{r}, t)}{d t}+\frac{M_{\perp}(\mathbf{r}, t)}{T_{2}}=0.
\end{equation}

\subsection{PINN for $T_{2}$ mapping}
Typical $T_2$ value of a single voxel (without variable $ \mathbf{r} $) is commonly estimated with the least square method. According to Bloch equation (\ref{2}), the signal model is:
\begin{equation}\label{3}
M_{\perp}(t)=M_{0} e^{-\frac{t}{T_{2}}},
\end{equation}
where $T_2$ and $M_0$ are parameters to be quantified. For the measured data  $M'_{\perp}\left(t_{s}\right)\in[0,~T]$ ($M'_{\perp}(t)=M_{\perp}(t)+\zeta(t)$, where $\zeta(t)$ is noise,  $t_{s}$ is a multiple time moment and $s \in\{1,2, \cdots, S\}$ in which $S $ is the number of samples (TEs). 
All $T_2$ value quantification is performed pixel-by-pixel. The least square methold also assumes that magnetization vector $M_{\perp}(t)$ is represented by the pixel intensity. Then, minimize the difference between two side of equation (\ref{3}) can obtain $T_2$ from multiple noisy samples by setting $t$ as multiple time of echo (TE). 

The proposed PINN is a neural network learning approach to approximate the Bloch Equation (\ref{2}) and obtain quantitative parameters. The network structure is a fully connected network with two hidden layers. The input of PINN is $ t $, where $ t $ is a set of $D$ discrete time points and S sample time points in $T_2$ parameter interval, and the network output is $\mathcal{N}(t)$. Then, $ \mathcal{N}(t) $ and  $M_{\perp}(t) $ should satisfy:
\begin{align}\label{4}
\frac{dM_{\perp}(t)}{dt}+\frac{M_{\perp}(t)}{T_2}=\frac{d\mathcal{N}(t)}{dt}+\frac{\mathcal{N}(t)}{T_2}=0.
\end{align}

Before presenting the loss function, we first define the residual terms required for its formulation:
\begin{align}\label{9}
\mathcal{R}_{d}(t)&:=\frac{d \mathcal{N}(t)}{d t}+\frac{\mathcal{N}(t)}{\hat{T}_{2}},\\\label{10}
\mathcal{R}_{s}(t)&:= \mathcal{N}(t) - M'_{\perp}(t),
\end{align}
where $\hat{T}_2$ is the estimation of $T_2$ by PINN. 
Both $\mathcal{N}(t)$ and $\hat{T}_{2}$ present errors than the real magnetization vector $M_{\perp}(t)$ and true $T_{2}$, respectively. Our objective is to minimize two residuals in Equations (\ref{9}) and (\ref{10}) in the loss function of PINN over time $t$. Accordingly, the first part is a physics-informed loss for the Bloch Equation (\ref{4}) as follows:
\begin{align}\label{5}
\mathcal{L}_{D,R} = \frac{1}{D}\sum^D_{d=1}\left\|\mathcal{R}_{d}(t_d)\right\|,
\end{align}
where $\|\cdot\|$ is the $l_2$ norm, $d\in\{1,2,\cdots,D\}$ is the number of discrete input time points $t$. The second part is the loss between the network output and the measured realistic cardiac MRI signal as follows:
\begin{align}\label{6}
\mathcal{L}_{S,R} = \frac{1}{S}\sum^S_{s=1}\|\mathcal{R}_{s}(t_s)\|.
\end{align}\\

Finally, the total loss function is a weighted sum of two parts:
\begin{align}\label{7}
\mathcal{L}=w_{D,R} \mathcal{L}_{D,R}+w_{S,R} \mathcal{L}_{S,R},
\end{align}
where $w_{S,R}$ and $w_{D,R}$ are two weights to balance importance of two terms.

\subsection{Theoretical analysis}
After quantitative estimation using PINN, the error bound theory can be derived by analyzing the relationship between the predicted parameters and the residuals. Besides using the standard residual expressions common to PINN (Equations (\ref{9}) and (\ref{10})), one needs to define a physical residual evaluated specifically at the sampling points $t$ in Equation (\ref{9}) as follows:
\begin{definition}[Physical residual with sample data]\label{def5}
Put the estimated solution $\mathcal{N}(t_{s})$ of sampled data (TEs) into the Bloch equation under $\hat{T}_{2} $ estimation, then the new physical parameter residual $\mathcal{R}_{sd}$ is
\begin{align}\label{13}
\mathcal{R}_{sd}(t_{s}):=\frac{d \mathcal{N}(t_{s})}{d t}+\frac{\mathcal{N}(t_{s})}{\hat{T}_{2}  },
\end{align}
where $s \in\{1,2, \cdots, S\}$, $S $ is the number of sampled data.
\end{definition}
Because the domain $[0,~T]$ of the Bloch equation is continuous along the t-axis of the coordinate system (Fig. \ref{figure1}(c)), the integrals of three residuals are:
\begin{align}\label{20}
&\left \| \mathcal{R}_{d}(t) \right \| _{p}^{p} =\int_{0}^{T}\left|\mathcal{R}_{d}(t)\right|^{p} \mathrm{d} t,\\
&\left \| \mathcal{R}_{s}(t) \right \| _{p}^{p} = \int_{0}^{T}\left|\mathcal{R}_{s}(t)\right|^{p} \mathrm{d} t ,\\
&\left \| \mathcal{R}_{sd}(t) \right \| _{p}^{p} = \int_{0}^{T}\left|\mathcal{R}_{sd}(t)\right|^{p} \mathrm{d} t,
\end{align}
where $||.||^p_p$ denotes the $l_{p}$ norm. 

For practical training on the PINN network, the continuity of the integration process $t$ is simulated at a discrete set of points in the time domain $[0,~T]$. Therefore, there exists an integration error when discretizing the integration. The discretizing process is illustrated in Fig. \ref{figure1}(c) and the integration error over all sampled points are presented in the following lemma.

\begin{Lemma}[Integration Error]\cite{drag1998}\label{lemma1}
For the residuals $\mathcal{R}_{d}(t),~\mathcal{R}_{s}(t)$ and $\mathcal{R}_{sd}(t)$ which we define, there exists constants $C_{D}$, $C_{S} ,~C_{SD} > 0$, such that their integration errors in the $l_{p}$ norm are bounded as
\begin{align}
&~~\left|\int_{0}^{T} \left|\mathcal{R}_{d}(t)\right|^{p} \mathrm{d} t-\sum_{d=1}^{D} w_{d}\left|\mathcal{R}_{d}\left(t_{d}\right)\right|^{p}\right| \leqslant C_{D} D^{-\alpha},\label{23}\\
&~~\left|\int_{0}^{T} \left|\mathcal{R}_{s}(t)\right|^{p} \mathrm{d} t-\sum_{s=1}^{S} w_{s}\left|\mathcal{R}_{s}\left(t_{s}\right)\right|^{p}\right| \leqslant C_{S} S^{-\alpha},\label{24}\\
&\left|\int_{0}^{T} \left|\mathcal{R}_{sd}(t)\right|^{p} \mathrm{d} t-\sum_{s=1}^{S} w_{s}\left|\mathcal{R}_{sd}\left(t_{s}\right)\right|^{p}\right| \leqslant C_{SD} S^{-\alpha},\label{25}
\end{align}
where 
\begin{align}
C_{D} = \left\{\begin{array}{ll}
\frac{T^{3}}{24D}\sum_{d=1}^{D}\left\|\mathcal{R}_{d}^{\prime \prime}\right\|_{\infty}, ~~~~~~~\text { if } f^{\prime \prime} \in L_{\infty}(x_{d-1}, x_{d}),  \nonumber\\
\frac{T^{2+\frac{1}{p}}}{8(2 pD +D)^{\frac{1}{p}}}\sum_{d=1}^{D}\left\|\mathcal{R}_{d}^{\prime \prime}\right\|_{q},  \text { if } f^{\prime \prime} \in L_{q}(x_{d-1}, x_{d}),\\
\frac{T^{2}}{8}\sum_{d=1}^{D}\left\|\mathcal{R}_{d}^{\prime \prime}\right\|_{1},  ~~~~~~~~~~\text { if } f^{\prime \prime} \in L_{1}(x_{d-1}, x_{d}),
\end{array}\right.
\end{align}
$C_{S}$ and $C_{SD}$ are same as $C_{D}$, $w_{d},~w_{s}$ are the discretization time steps, $\alpha=2$ when using the midpoint rule to approximate the integral, and $\alpha=3$ when using the trapezoidal rule. The second derivative $\mathcal{R}_{d}^{\prime \prime}$ of $\mathcal{R}_{d}$ is obtained via a central finite-difference numerical differentiation method.
\end{Lemma}

\begin{definition}[Training Error]\label{def9}
For the the residuals $\mathcal{R}_{d}(t)$, $\mathcal{R}_{s}(t)$ and $\mathcal{R}_{sd}(t)$ of training PINN, we can define the training errors as follows:
\begin{align}
\mathcal{E}_{D, R}:=\left(\sum_{d=1}^{D} w_{d}\left|\mathcal{R}_{d}\left(t_{d}\right)\right|^{p}\right)^{\frac{1}{p}}, 
\end{align}
where $\mathcal{E}_{S, R}$ and $\mathcal{E}_{SD, R}$ are same as $\mathcal{E}_{D, R}$, and $w_{d},~w_{s}$ are the discrete time step, and when the step intervals are equal, $\bigtriangleup t = w_{s} = t_{s} -t_{s-1},~1\le s\le S$ (or $\bigtriangleup t = w_{d} = t_{d} -t_{d-1},~1\le d\le D$).
\end{definition}

Similarly, according to Lemma 1, for discrete magnetic intensity $M_{\perp}(t_{s})$, the integration of the noiseless data $M_{\perp}(t_{s})$ can be bounded under the $C_{M}> 0$ as
\begin{align}
\left|\int_{0}^{T} \left|M_{\perp}(t)\right|^{p} \mathrm{d} t-\sum_{s=1}^{S} w_{s}\left|M_{\perp}\left(t_{s}\right)\right|^{p}\right| \leqslant C_{M} S^{-\alpha},\label{}
\end{align}
and we then let
\begin{align}
\mathcal{I}_{M}:=\left(\sum_{s=1}^{S} w_{s}\left|M_{\perp}\left(t_{s}\right)\right|^{p}\right)^{\frac{1}{p}}.
\end{align}

\begin{definition}[Noisy Data]\label{def11}
If the noisy data $M'_{\perp}(t)$ and the noiseless data $M_{\perp}(t)$ have a noise $\zeta(t)$ such that
\begin{align}\label{51}
M'_{\perp}(t)=M_{\perp}(t)+\zeta(t),
\end{align}
the definition and the assumption are equally applicable for
\begin{align}\label{}
\mathcal{R}_{s}(t):=\mathcal{N}(t) - M'_{\perp}(t).\label{52}
\end{align}
And let
\begin{align}
\mathcal{I}_{No}:=\left(\sum_{s=1}^{S} w_{s}\left|\zeta\left(t_{s}\right)\right|^{p}\right)^{\frac{1}{p}}.
\end{align}
\end{definition}

With these established definitions and lemmas, we can derive the error bound theory for PINN-based quantitative estimation. All $T_2$ value quantification is performed pixel-by-pixel. The following theory assume that magnetization vector $M_{\perp}(t)(M_{real}(t))$ is represented by the real pixel intensity. Definitions of main variables are shown in Fig. \ref{figure1}(c). The loss of sampled data (Fig. \ref{figure1}(c1)) is defined as the sum of residuals between the experimentally sampled data (sampled pixel intensity $M'_{\perp}(t)$ ($M'_{sample}(t)$)) and PINN-fitted magnetization vector $\mathcal{N}(t)$ for all sampling time $t_s (s=1,2....,S)$. Fig. \ref{figure1}(c2) presents the noise $\zeta(t)$ of sampled data, which satisfies $M_{\perp}(t) = M'_{\perp}(t)-\zeta(t)$. The sum of the noise $\zeta(t)$ is denoted as $\mathcal{I}_{No} $ $(\mathcal{I}_{Noise})$.

Theorem \ref{thm2} provides the theoretical upper bound of estimation error $|T_{2} - \hat{T}_{2}|/T_{2}$ for noisy case, where $T_2$ and $\hat{T}_{2}$ is the ground-truth and estimated parameter, respectively. 

\begin{theorem}[Error bound of $T_{2}$ estimation on noisy data]\label{thm2}
Let $\hat{T}_{2} $ be the PINN estimation of the ground truth value $T_{2}$, $\left \{ M'_{\perp}(t_{s})  \right \}_{s=1}^{S}$ be the noisy sampled data of Bloch equation, and the training set  be $\mathcal{T}_{sample} =\left \{ t_{s}  \right \}_{s=1}^{S} $, then the estimate error of $T_2$ satisfies:
\begin{align}\label{2.1}
 \frac{|T_{2}-\hat{T}_{2}|}{T_{2}  } \le & \frac{1}{\mathcal{I}_{M} - C_{M}^{\frac{1}{p}} S^{-\frac{\alpha}{p}}}\Big(\mathcal{E}_{S, R}  +\hat{T}_{2}\mathcal{E}_{SD, R} + \mathcal{I}_{No}   \nonumber\\
&+ \hat{T}_{2} \left | \bigtriangleup  \mathcal{R}_{s}(t_{s})\right | + \hat{T}_{2} \left | \bigtriangleup  \zeta(t_{s})\right | +C_{S}^{\frac{1}{p}} S^{-\frac{\alpha}{p}} \nonumber\\
& +C_{No}^{\frac{1}{p}} S^{-\frac{\alpha}{p}} 
 + \hat{T}_{2}C_{SD}^{\frac{1}{p}} S^{-\frac{\alpha}{p}}\Big),
\end{align}
where $\left | \bigtriangleup  \mathcal{R}_{s}(t_{s})\right | = \left |\max\limits_{t_{s}} \mathcal{R}_{s}(t_{s})-\min\limits_{t_{s}} \mathcal{R}_{s}(t_{s})\right |, \left | \bigtriangleup  \zeta(t_{s})\right | = \left |\max\limits_{t_{s}} \zeta(t_{s})-\min\limits_{t_{s}} \zeta(t_{s})\right |,t_{s}\in [0,T].$
\end{theorem}

\begin{proof}
~~According to (4) and (10), we know that
\begin{align}
&\frac{d M_{\perp}(t)}{d t}+\frac{M_{\perp}(t)}{T_{2}}=0,\label{33}\\ \label{A11}
&\frac{d \mathcal{N}(t)}{d t}+\frac{\mathcal{N}(t)}{\hat{T}_{2} }=\mathcal{R}_{sd}(t).
\end{align}
~~We subtract equation (\ref{33}) from equation (\ref{A11}), then we have
\begin{align}\label{}
\frac{d [\mathcal{N}(t)-M_{\perp}(t)]}{d t}+\frac{\mathcal{N}(t)}{\hat{T}_{2} } - \frac{M_{\perp}(t)}{T_{2} } =  \mathcal{R}_{sd}(t),
\end{align}
adding and subtracting $\frac{M_{\perp}(t)}{T_{2}}$ on the left,
\begin{align}\label{}
&\frac{d [\mathcal{N}(t)-M_{\perp}(t)]}{d t}+\frac{[\mathcal{N}(t)-M_{\perp}(t)]}{\hat{T}_{2} } - \frac{M_{\perp}(t)}{\hat{T}_{2}} + \frac{M_{\perp}(t)}{T_{2}} \nonumber\\
= & \mathcal{R}_{sd}(t).
\end{align}
so,
\begin{align}\label{36}
&\frac{d [\mathcal{N}(t)-M_{\perp}(t)]}{d t}+\frac{[\mathcal{N}(t)-M_{\perp}(t)]}{\hat{T}_{2} }\nonumber\\
=  & \mathcal{R}_{sd}(t) + \frac{M_{\perp}(t)}{T_{2}} - \frac{M_{\perp}(t)}{\hat{T}_{2}}.
\end{align}
~~Next, from equation (\ref{36}),
\begin{align}\label{37}
\frac{M_{\perp}(t)}{T_{2}} - \frac{M_{\perp}(t)}{\hat{T}_{2}} = M_{\perp}(t)\left ( \frac{1}{T_{2}} - \frac{1}{\hat{T}_{2}} \right ) \nonumber\\
=\left ( \frac{\hat{T}_{2}-T_{2}}{T_{2}\hat{T}_{2}}  \right )M_{\perp}(t).
\end{align}
~~According to (\ref{36}) and (\ref{37}), we know that
\begin{align}\label{58}
&\frac{d [\mathcal{N}(t)-M_{\perp}(t)]}{d t}+\frac{[\mathcal{N}(t)-M_{\perp}(t)]}{\hat{T}_{2} } \nonumber\\
= & \mathcal{R}_{sd}(t) + \frac{M_{\perp}(t)}{T_{2}} - \frac{M_{\perp}(t)}{\hat{T}_{2}}\nonumber\\
=  & \mathcal{R}_{sd}(t) + \left ( \frac{\hat{T}_{2}-T_{2}}{T_{2}\hat{T}_{2}}  \right )M_{\perp}(t),
\end{align}
and through (\ref{51}) and (\ref{52})
\begin{align}\label{}
\mathcal{R}_{s}(t) = \mathcal{N}(t) - M'_{\perp}(t)= \mathcal{N}(t) - (M_{\perp}(t)+\zeta(t)),
\end{align}
we have that
\begin{align}\label{}
&\mathcal{N}(t) - M_{\perp}(t) = \mathcal{R}_{s}(t) + \zeta(t),\\
&M_{\perp}(t)=M'_{\perp}(t) - \zeta(t),
\end{align}
so,
\begin{align}\label{62}
 & \frac{d [\mathcal{N}(t)-M_{\perp}(t)]}{d t}+\frac{[\mathcal{N}(t)-M_{\perp}(t)]}{\hat{T}_{2} } \nonumber\\
  = & \frac{d \mathcal{R}_{s}(t)}{d t}+\frac{\mathcal{R}_{s}(t)}
    {\hat{T}_{2} } + \frac{d \zeta(t)}{d t}+\frac{\zeta(t)}{\hat{T}_{2} }.
\end{align}
~~Then, combining (\ref{58}) and (\ref{62}), it would be
  \begin{align}\label{}
& \frac{d \mathcal{R}_{s}(t)}{d t}+\frac{\mathcal{R}_{s}(t)}
    {\hat{T}_{2} } + \frac{d \zeta(t)}{d t}+\frac{\zeta(t)}{\hat{T}_{2} } \nonumber\\
    = &\mathcal{R}_{sd}(t) + \left ( \frac{\hat{T}_{2}-T_{2}}{T_{2}\hat{T}_{2}}  \right )M_{\perp}(t) \nonumber\\
 =  & \mathcal{R}_{sd}(t) + \left ( \frac{\hat{T}_{2}-T_{2}}{T_{2}\hat{T}_{2}}  \right )(M'_{\perp}(t) - \zeta(t)).
\end{align}
~~Then, we can obtain
\begin{align}\label{}
&\left ( \frac{\hat{T}_{2}-T_{2}}{T_{2}\hat{T}_{2}}  \right )(M'_{\perp}(t) - \zeta(t)) \nonumber\\
= & \frac{d \mathcal{R}_{s}(t)}{d t}+\frac{\mathcal{R}_{s}(t)}{\hat{T}_{2} } + \frac{d \zeta(t)}{d t} +\frac{\zeta(t)}{\hat{T}_{2} } - \mathcal{R}_{sd}(t).
\end{align}
~~Next, integrating both sides of the equation
\begin{align}\label{}
&\left ( \frac{\hat{T}_{2}-T_{2}}{T_{2}\hat{T}_{2}}  \right )\left ( \int_{0}^{T}M'_{\perp}(t)\mathrm{d} t - \int_{0}^{T}\zeta(t)\mathrm{d} t \right )\nonumber\\
=& \int_{0}^{T}\frac{d \mathcal{R}_{s}(t)}{d t} \mathrm{d} t + \int_{0}^{T}\frac{\mathcal{R}_{s}(t)}{\hat{T}_{2} }\mathrm{d} t + \int_{0}^{T}\frac{d \zeta(t)}{d t} \mathrm{d} t \nonumber\\ 
&+ \int_{0}^{T}\frac{\zeta(t)}{\hat{T}_{2} }\mathrm{d} t - \int_{0}^{T}\mathcal{R}_{sd}(t)\mathrm{d} t\nonumber\\
=& \left[ \mathcal{R}_{s}(t) \right ]\big| _{0}^{T} + \int_{0}^{T}\frac{\mathcal{R}_{s}(t)}{\hat{T}_{2} }\mathrm{d} t + \left[ \zeta(t) \right ]\big| _{0}^{T} + \int_{0}^{T}\frac{\zeta(t)}{\hat{T}_{2} }\mathrm{d} t \nonumber\\
& - \int_{0}^{T}\mathcal{R}_{sd}(t)\mathrm{d} t,
\end{align}
so
\begin{align}\label{65}
&\left | \frac{\hat{T}_{2}-T_{2}}{T_{2}}  \right |   
\le   \frac{1}{\left (   \int_{0}^{T}\left | M'_{\perp}(t) -\zeta(t) \right |^{p} \mathrm{d} t\right )^{\frac{1}{p} }  } \Bigg(   \hat{T}_{2} \left |\left[ \mathcal{R}_{s}(t) \right ] \big|_{0}^{T}\right |   \nonumber\\
&   + \left (   \int_{0}^{T}\left |\mathcal{R}_{s}(t) \right |^{p} \mathrm{d} t\right )^{\frac{1}{p} }  +  \hat{T}_{2} \left |\left[ \zeta(t) \right ]\big| _{0}^{T}\right | + \left (   \int_{0}^{T}\left |\zeta(t) \right |^{p} \mathrm{d} t\right )^{\frac{1}{p} } \nonumber\\
&  +  \hat{T}_{2} \left (   \int_{0}^{T}\left |\mathcal{R}_{sd}(t) \right |^{p} \mathrm{d}t\right )^{\frac{1}{p} } \Bigg).
\end{align}\\

Due to $t_{s}\in  \left [ 0,T \right ] $ for $1\le s\le S$ such that $[t_{1},t_{S}]\subseteq \left [ 0,T \right ] $, so we have
\begin{align}\label{43}
\left |\left[ \mathcal{R}_{s}(t) \right ] _{0}^{T}\right |  = &\left |\mathcal{R}_{s}(T)- \mathcal{R}_{s}(0)\right |  \\ \nonumber
\le & \left |\max_{t_{s}\in [0,T]} \mathcal{R}_{s}(t_{s})-\min_{t_{s}\in [0,T]} \mathcal{R}_{s}(t_{s})\right | = \left | \bigtriangleup  \mathcal{R}_{s}(t_{s})\right |.
\end{align}
and
\begin{align}\label{66}
\left |\left[ \zeta(t) \right ]\big| _{0}^{T}\right | & = \left |\zeta(T)- \zeta(0)\right |  \\  \nonumber
&\le  \left |\max_{t_{s}\in [0,T]} \zeta(t_{s})-\min_{t_{s}\in [0,T]} \zeta(t_{s})\right | = \left | \bigtriangleup  \zeta(t_{s})\right |,
\end{align}

On the $l^{p} $ norm, $p$ always satisfies $p\ge 1$. Thus, by applying Jensen's inequality, which is
\begin{align}
\left ( a^{p}  +b^{p} \right ) ^{\frac{1}{p} } \le a+b,~p\ge 1,
\end{align}
and using the result of (18) and (19) that
\begin{align}\label{A20}
&\int_{0}^{T}\left |\mathcal{R}_{s}(t) \right |^{p} \mathrm{d} t \le \mathcal{E}_{S, R}^{p} +C_{S}S^{-\alpha},\\
&\int_{0}^{T}\left |\mathcal{R}_{sd}(t) \right |^{p} \mathrm{d} t \le \mathcal{E}_{SD, R}^{p} +C_{SD}S^{-\alpha},
\end{align}
we have
\begin{align}\label{}
& \left (   \int_{0}^{T}\left |\mathcal{R}_{s}(t) \right |^{p} \mathrm{d} t\right )^{\frac{1}{p} }  \le \mathcal{E}_{S, R} +C_{S}^{\frac{1}{p}} S^{-\frac{\alpha}{p}}, \label{47}\\
& \left (   \int_{0}^{T}\left |\mathcal{R}_{sd}(t) \right |^{p} \mathrm{d} t\right )^{\frac{1}{p} }  \le \mathcal{E}_{SD, R} +C_{SD}^{\frac{1}{p}} S^{-\frac{\alpha}{p}}.\label{48}.
\end{align}
~~Similarly,
\begin{align}\label{A55}
\left (   \int_{0}^{T}\left |\zeta(x) \right |^{p} \mathrm{d} x\right )^{\frac{1}{p} }  &\le \mathcal{I}_{No, R} +C_{No}^{\frac{1}{p}} S^{-\frac{\alpha}{p}},\\ \label{A56}
\left (   \int_{0}^{T}\left |M'_{\perp}(t) - \zeta(t) \right |^{p} \mathrm{d} t\right )^{\frac{1}{p} } & =\left (   \int_{0}^{T}\left |M_{\perp}(t) \right |^{p} \mathrm{d} t\right )^{\frac{1}{p} }\nonumber\\
&\ge \mathcal{I}_{M} - C_{M}^{\frac{1}{p}} S^{-\frac{\alpha}{p}}.
\end{align}
~~Using (\ref{43}), (\ref{66}), (\ref{47}), (\ref{48}), (\ref{A55}) and (\ref{A56}), we can finally draw the conclusion from (\ref{65})
\begin{align}\label{72}
 \left| \frac{T_{2}-\hat{T}_{2}  }{\hat{T}_{2}  }  \right| 
 \le & \frac{1}{\mathcal{I}_{M} - C_{M}^{\frac{1}{p}} S^{-\frac{\alpha}{p}}}\Big(\mathcal{E}_{S, R}  +\hat{T}_{2}\mathcal{E}_{SD, R} +\mathcal{I}_{No} \nonumber\\
& + \hat{T}_{2}\left | \bigtriangleup  \mathcal{R}_{s}(t_{s})\right | + \hat{T}_{2}\left | \bigtriangleup  \zeta(t_{s})\right | + C_{S}^{\frac{1}{p}} S^{-\frac{\alpha}{p}}  \nonumber\\
& + C_{No}^{\frac{1}{p}} S^{-\frac{\alpha}{p}} + \hat{T}_{2}C_{SD}^{\frac{1}{p}} S^{-\frac{\alpha}{p}}\Big).
\end{align}
\end{proof}

This theorem points out the maximal error could be reached by PINN and has important interpretations:

1) Error bound does not depend on the ground-truth $T_{2}$ value but depends on the estimated $\hat{T}_{2}$, providing an accessible way to evaluate the maximal error without gold standard or ground truth. The reason is $T_{2}$ is not presented in the right side of the inequality. In addition, an increasing $\hat{T}_{2}$ will slow signal decay in the numerator, but also will enlarge $M_{\perp}(t)$ and its integral $\mathcal{I}_{M}$ in the denominator, so it may lead to a partial compensation of their contributions to the overall error bound. However, due to their nonlinear and coupled relationship, a larger $\hat{T}_{2}$ increase the estimation error bound but a larger $\mathcal{I}_{M}$ decrease the estimation error bound in a complex manner.

2) Increasing the number of sampled data $S$ will reduce parameter estimation error bound. The reasons is $S$ appears in the bound as an exponential function. For example, $S^{-\frac{\alpha}{p}}={S^{-2}}$ if the integral discretization follows the midpoint rule, i.e., $\alpha=2$, and the $l_{p}$ norm has $p=1$. Besides, a larger $S$ leads to improved curve fitting, resulting in smaller residual sums of $\mathcal{E}_{S, R}$ and $\mathcal{E}_{SD, R}$ at each echo time.  

3) Noise is an essential factor to reduce the estimation error bound. As noise $\zeta(t)$ increases, the sum $\mathcal{I}_{No}$ and $\Delta \zeta(t_s)$ in the error bound also increases, leading to a larger estimation error bound. The deviation of sampled data $M'_{\perp}(t)$ from the real $M_{\perp}(t)$ also grows with noise, degrading fitting performance.

\begin{Corollary}[Generalization error of noisy data]\label{cor1}
The estimate of noisy data on the generalization error satisfies
\begin{align}\label{96}
\mathcal{E}_{G}(t)&=\left\|M_{\perp}(t)-\mathcal{N}(t)\right\|\nonumber\\
&\le  e^{-\frac{t}{\hat{T}_{2}  } } \Big( C_{0}^{\ast }  + \tilde{ \mathcal{E}} _{sd, R}+\tilde{ \mathcal{E}} _{d, R} + \left |\bigtriangleup e^{\frac{t_{s}}{\hat{T}_{2} } } \mathcal{R}_{s}(t_{s})   \right | \nonumber\\
& ~~~+ \left |\bigtriangleup e^{\frac{t_{s}}{\hat{T}_{2} } } \zeta(t_{s})   \right | + C_{sd}^{\frac{1}{p}} s^{-\frac{\alpha}{p}} + C_{d}^{\frac{1}{p}} d^{-\frac{\alpha}{p}} \Big),
\end{align}
where $ \left | \bigtriangleup \mathrm{e}^{\frac{t_{s}}{\hat{T}_{2} } } \mathcal{R}_{s}(t_{s})\right | = \left |\max\limits_{t_{s}} \mathrm{e}^{\frac{t_{s}}{\hat{T}_{2} } }\mathcal{R}_{s}(t_{s})-\min\limits_{t_{s}} \mathrm{e}^{\frac{t_{s}}{\hat{T}_{2} } }\mathcal{R}_{s}(t_{s})\right |$ and $ \left | \bigtriangleup \mathrm{e}^{\frac{t_{s}}{\hat{T}_{2} } } \zeta(t_{s})\right | = \left |\max\limits_{t_{s}} \mathrm{e}^{\frac{t_{s}}{\hat{T}_{2} } }\zeta(t_{s})-\min\limits_{t_{s}} \mathrm{e}^{\frac{t_{s}}{\hat{T}_{2} } }\zeta(t_{s})\right |,t_{s}\in [0,t].$
\end{Corollary}
Proof: derivation is similar to Theorem 2.

Corollary \ref{cor1} provides the generalization error bound that measures the largest difference between the PINN-fitted curve $\mathcal{N}(t)$ and the real curve $M_{\perp}(t)$ ($M_{real}(t)$) in the continuous-time domain. Although the real curve $M_{\perp}(t)$ is unknown throughout the timeline, this corollary provides an upper bound curve over entire continuous-time domain, allowing to estimate potential error at any time $t$. Thus, although ground-truth $T_2$ map is not unknown in vivo imaging, it is possible to validate the accuracy of estimated $T_2$ by acquiring more samples (TEs), obtaining a better $T_2$ map to generate MRI images, and then measuring the image error to the physically acquired image.

\section{RESULTS}
In this section, the quantitative and theoretical analysis is validated through three distinct datasets: 1) numerical simulations on a cardiac model, 2) experimental measurements from a industrial standard qMRI phantom, and 3) clinical scans from patient cohorts. For simulated data, theoretical quantification errors are validated from multiple controlled factors, such as the the signal-to-noise ratio (SNR) and the number of samples, i.e. the number of TE. Additionally, we compare $T_{2}$ values estimated by PINN with the conventional least square method, the clinical gold standard, and with DFN; PINN achieves nearly identical accuracy to least square and outperforms DFN. In phantom data, PINN provides more precise $T_{2}$ estimation than both least square and DFN methods in the myocardial $T_{2}$ range. For in vivo patient data, where ground truth is unavailable, our theoretical error bound indicate that PINN achieves low quantitative errors, and multiple radiologists assess the quality of the resulting $T_{2}$ maps to evaluate diagnostic reliability.

\subsection{Numerical cardiac experiments}
\begin{figure}[]
  \centering
  \includegraphics[width=0.8\textwidth]{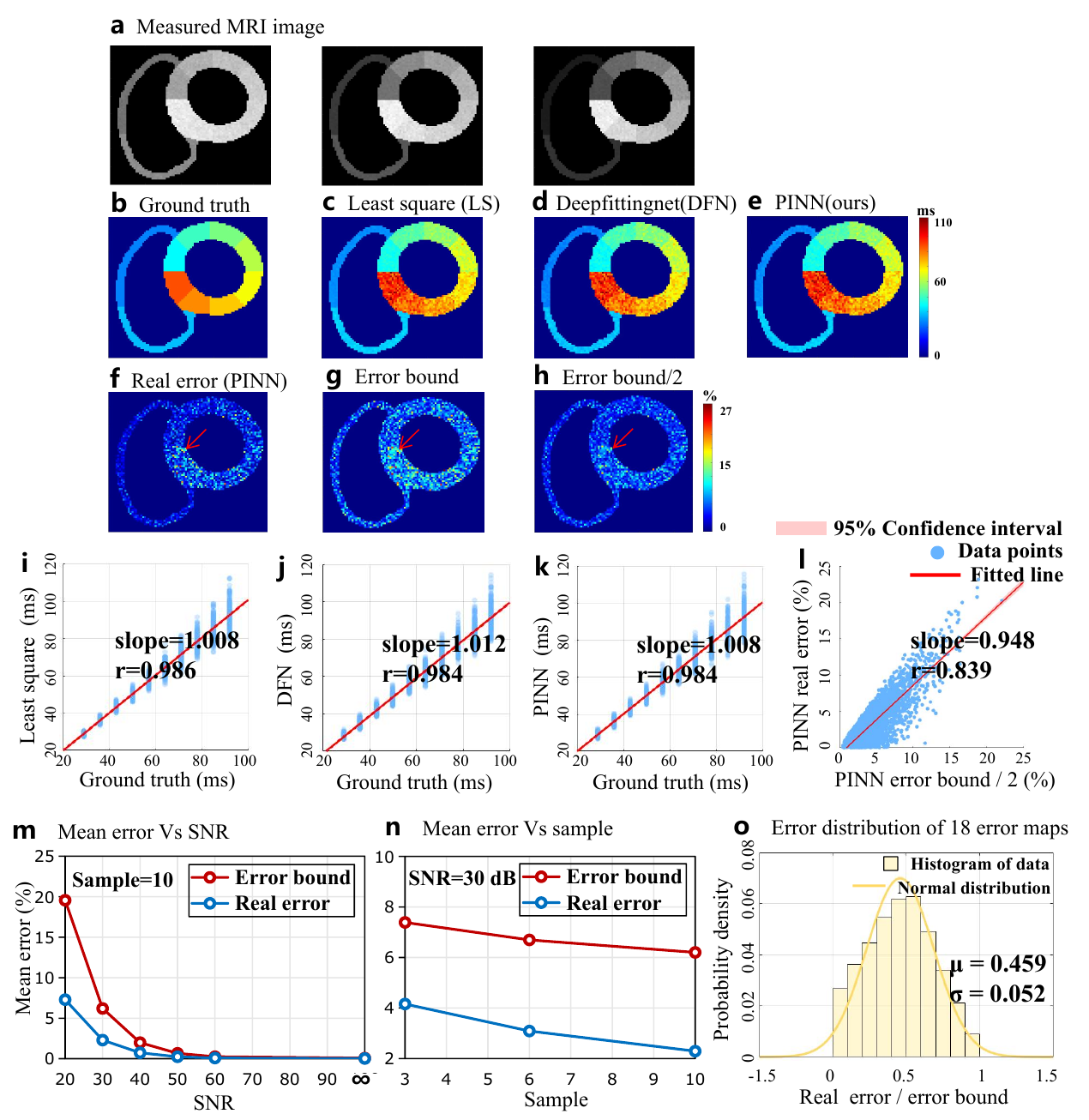}
  \caption{$T_{2}$ mapping and error bound of PINN of the numerical cardiac model. 3 echo times (TEs) were selected from 30 ms to 90 ms with an interval of 30 ms (TE 1 = 30 ms, TE 2 = 60 ms, TE 3 = 90 ms). The signal-to-noise ratio (SNR) = 30 dB. (i)-(k) are $T_{2}$ correlation between three methods (least square, DFN and PINN) and groud truth. (l) is correlation between real error (PINN) and theoretical error bound/2. (m,n) Estimation error of PINN versus the SNR and the number of samples: (m) SNR, (n) the number of samples (TEs), (o) an approximate normal distribution of the ratio (real error/theoretical error bound in 30 - 100 ms pixels). Note: Mean error in (m) and (n) is an average value of all an entire error map. The statics in (o) is computed over pixels that have $T_2$ value between 30 - 100 ms since this is an possible range for normal and acute myocardial infarction tissues.}
  \label{figure2}
\end{figure}

\begin{figure}[]
  \centering
  \includegraphics[width=0.8\textwidth]{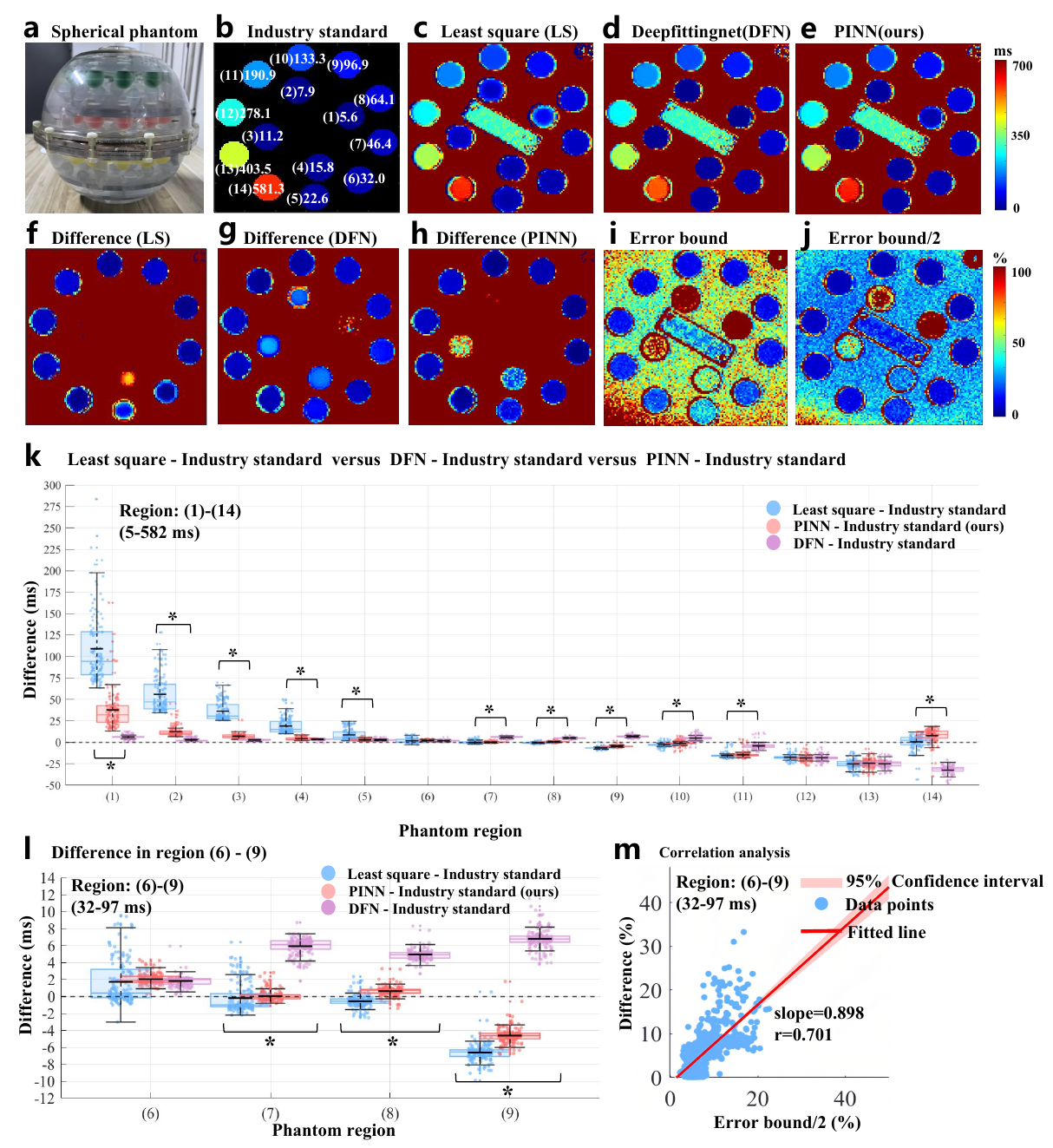}
  \caption{$T_{2}$ estimation error of a water phantom. (a) Phantom with 14 spheres. (b) Reference $T_{2}$ value (Industry standard) in 14 regions. (c)-(e) are estimated $T_{2}$ by least square, DFN and PINN, respectively. (f)-(h) are the difference of $T_{2}$ value between estimation and industry standard reference value. (i) and (j) are pixel-wise theoretical error bound and its half value. (k) and (l) are the difference between estimated $T_{2}$ and reference values in all regions and region (6)-(9), respectively. (m) Correlation analysis for theoretical error bound/2 and empirical difference. Note: 10 TEs were selected from 15 ms to 150 ms with an interval of 15 ms on a 3T NeuMR Universal scanner. $*$ represents the significant difference between two methods (p $<$ 0.05).}
  \label{figure66}
\end{figure}

A numerical cardiac model is synthesized with given $T_{2}$-values (Fig. \ref{figure2}). In the following, we will verify the $T_{2}$ estimation error bound and how SNR (The definition is in Appendix B) and the number of samples (TEs) affect this error.

The numerical cardiac model is constructed on a 80×90 pixel grid with 10 distinctive $T_{2}$ values (ground truth) to simulate 10 myocardial regions. Using the Bloch equation, we generated simulated $T_{2}$-weighted images at various echo times (TEs). Three TEs were selected from 30 ms to 90 ms with an interval of 30 ms; Six TEs were selected from 15 ms to 90 ms with an interval of 15 ms; Ten TEs were selected from 10 ms to 100 ms with an interval of 10 ms. For each, we generated both noiseless images and noisy versions with Gaussian noise at SNR levels of 60 dB, 50 dB, 40 dB, 30 dB, and 20 dB. $T_{2}$ maps are estimated by PINN and least square. PINN parameters include  $D$ = 1001, $S$ = 3 (6 / 10), $w_{D,R}$ = 0.01 and $w_{S,R}$ = 1 and theoretical parameters include $p$ = 1 and $\alpha=2$. The relative percentage error = $|T_2 - \hat{T}_2|/T_2$ is measured since the ground-truth value $T_2$ is available in this simulation.

A high consistency on $T_{2}$ value was observed across all three methods (least square, DFN and PINN) when compared with the ground truth (similar $T_{2}$ maps in Fig. \ref{figure2}(c)-(e)). The least-square method as the gold standard achieved the highest quantitative accuracy, with the strongest correlations and slopes closest to unity ($r=0.986$ and slope = 1.008 in Fig. \ref{figure2}(i)). PINN achieved comparable performance with slightly lower but still high accuracy ($r=0.984$ and slope = 1.008 in Fig. \ref{figure2}(k)). DFN showed good agreement as well, though with marginally larger deviations in slope ($r=0.984$ and slope = 1.012 in Fig. \ref{figure2}(j)). Moreover, the real error of PINN (Fig. \ref{figure2}(f)) looks similar to its theoretical bound/2 (Fig. \ref{figure2}(h)). This observation is further confirmed by the Pearson correlation coefficient $r=0.839$ in Fig. \ref{figure2}(l).

How SNR and the number of samples (TEs) affect the $T_2$ estimation error of PINN is analyzed in Fig. \ref{figure2}(m,n). Both theoretical and real error decrease fast with increasing SNR (Fig. \ref{figure2}(m)). Increasing the number of samples (TEs) also reduces errors (Fig. \ref{figure2}(n)). In all 18 cases (the number of TE is 3, 6 and 10; noise-free and noisy condition under SNR = 60 dB, 50 dB, 40 dB, 30 dB and 20 dB), the mean value of the ratio (real error/theoretical error bound) is 0.459 and its standard deviation is 0.052 under the normal distribution (p-value = 0.56 $>$ 0.05) (Fig. \ref{figure2}(o)) \cite{lilliefors1967ks, habibzadeh2024normality}, suggesting that theoretical error bound/2 can serve as a good estimate of the real error.

\subsection{Realistic data experiments}

Two realistic datasets are acquired from: (1) water phantom, (2) patients with acute myocardial infarction (AMI). PINN parameters are set as $D$ = 1001, $S$ = 10 for phantom (9 / 3 for patients), $w_{D,R}$ = 0.01 and $w_{S,R}$ = 1 for all experiments below. All patient-related data in this research have been approved by the ethics committee.

\subsubsection{Phantom}
A water phantom (Fig. \ref{figure66}(a)), an industry standard, of $T_{2}$ values (Fig. \ref{figure66}(b)) was adopted as the reference. Its standard $T_{2}$ value of 14 spheres was measured and averaged across different scanners. 

For the reference low $T_{2}$-values  (5.6 ms, 7.9 ms, 11.2 ms and 15.8 ms), least square presents much larger difference than PINN does and DFN shows less difference than other two methods (region (1)-(5) in Fig. \ref{figure66}(k)). For the reference high values (133.3 ms, 190.9 ms, 278.1 ms, 403.5 ms and 581.3 ms), both least square and PINN methods achieve comparable smaller differences, while DFN method is not as good as them (region (10)-(14) in Fig. \ref{figure66}(k)). For the reference moderate $T_{2}$-values (32.0 ms, 46.4 ms, 64.1 ms and 96.9 ms), least square performs sligtly lower difference than PINN (Region (8) in Fig. \ref{figure66}(l)) and DFN shows sligtly lower difference than PINN (Region (6) in Fig. \ref{figure66}(l)). However, PINN outperforms least square (Region (6) and (7) in Fig. \ref{figure66}(l)) and DFN (Region (7)-(9) in Fig. \ref{figure66}(l)) on much smaller variations and smaller difference than least square (Region (9) in Fig. \ref{figure66}(l)). Since the $T_{2}$ values of the human heart generally fall within 30 – 100 ms \cite{giri2009}, the improved accuracy of PINN suggests that it may provide better $T_{2}$ quantification for real cardiac applications. Besides, theoretical error bound/2 fits well with the practical difference from industry standard (Pearson correlation coefficient $r$ = 0.701 in Fig. \ref{figure66}(m)).

\subsubsection{Patients with acute myocardial infarction (AMI)}
  Here, we explore the potential of $T_{2}$ mapping in contrast injection-free imaging. $T_{2}$ mapping is noninvasive, contrast-free and cost-effective. It has been shown to partially or fully replace late gadolinium enhancement (LGE), particularly in AMI. Although LGE remains the gold standard for detecting myocardial fibrosis and necrosis, it requires contrast injection, making it invasive, costly, and inappropriate for patients with gadolinium allergies or renal impairment \cite{ghaffari2025, obrien2022}. Thus, $T_{2}$ mapping has important clinical value. Even though, LGE images were used as a reference to validate the lesions identified by $T_{2}$ mapping and enhance the rigor and clinical relevance of our findings. 

  Patient data includes black-blood dataset and bright-blood dataset. Black-blood dataset is acquired at 9 TEs from 48 AMI patients(mean age: 72$\pm$16). Bright-blood images is acquired at 3 TEs from 46 AMI patients (mean age: 58$\pm$10). 

Two representative $T_{2}$ maps are shown in Fig. \ref{figure10}(a). From two type images, PINN, DFN and least square methods estimate highly close $T_{2}$ values (1 ms difference for bright-blood images; PINN differs from least square by $\sim$1 ms, while DFN differs by $\sim$5 ms for bright-blood images) at infarction lesions. At these regions, the error bound of PINN is 5\% for black-blood images or 7\% for black-blood images. These observations means that the maximal $T_{2}$ estimation error of PINN is $74 \text{ms}*5\%=3.7 \text{ms}$ or $75 \text{ms}*7\%=5.25 \text{ms}$, which are very low. Thus, the theoretical error bound/2 is remarkably low in lesions. In contrast, in the healthy myocardium (red circles), $T_{2}$ values of DFN are outside the normal range ($\sim$60–70 ms in black-blood images and $\sim$30–40 ms in bright-blood images), while PINN and least square methods remain within physiological limits. Overall, PINN provides accurate and reliable $T_{2}$ estimation in both infarcted and non-infarcted regions.

\begin{figure}[]
  \centering
  \includegraphics[width=0.8\textwidth]{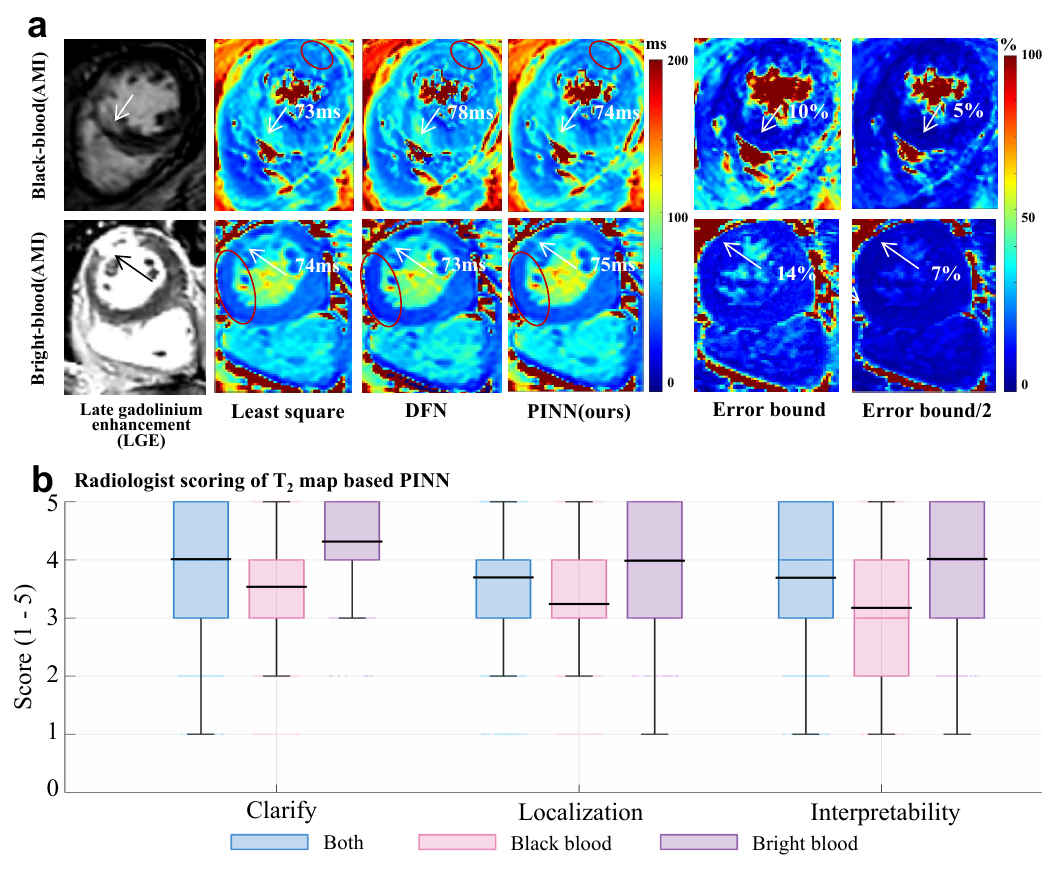}
  \caption{Representative $T_2$ maps of AMI patients and radiologist evaluation on all patient $T_2$ maps. (a) $T_{2}$ mapping and error bound for AMI patients, where the first row is estimated from black-blood images (9 TEs were selected from 8.845 ms to 79.605 ms with an interval of 8.845 ms on a Philips Ingenia 3T scanner) while the second row is estimated from bright-blood images (3 TEs were set as 0 ms, 30 ms and 55 ms on a United Imaging UMR790 3T scanner). (b) Radiologist scoring of $T_2$ map estimated by PINN. Three image metrics (Clarify, Localization, and Interpretability) were set on a 5-point scale (5 = excellent, 4 = good, 3 = moderate, 2 = low, 1 = poor, minimal interval=1). Box plots show the median (colored line within the box), the mean value (black horizontal line in the box), the 25th and 75th percentiles (box edges),  and whiskers for data range. Individual scores are shown as jittered circles. The scoring dataset included images of $T_{2}$ quantification results from 3–7 short-axis cardiac slices of 94 AMI patients, with corresponding LGE images provided as references.}
  \label{figure10}
\end{figure}

\begin{figure}[]
  \centering
  \includegraphics[width=0.8\textwidth]{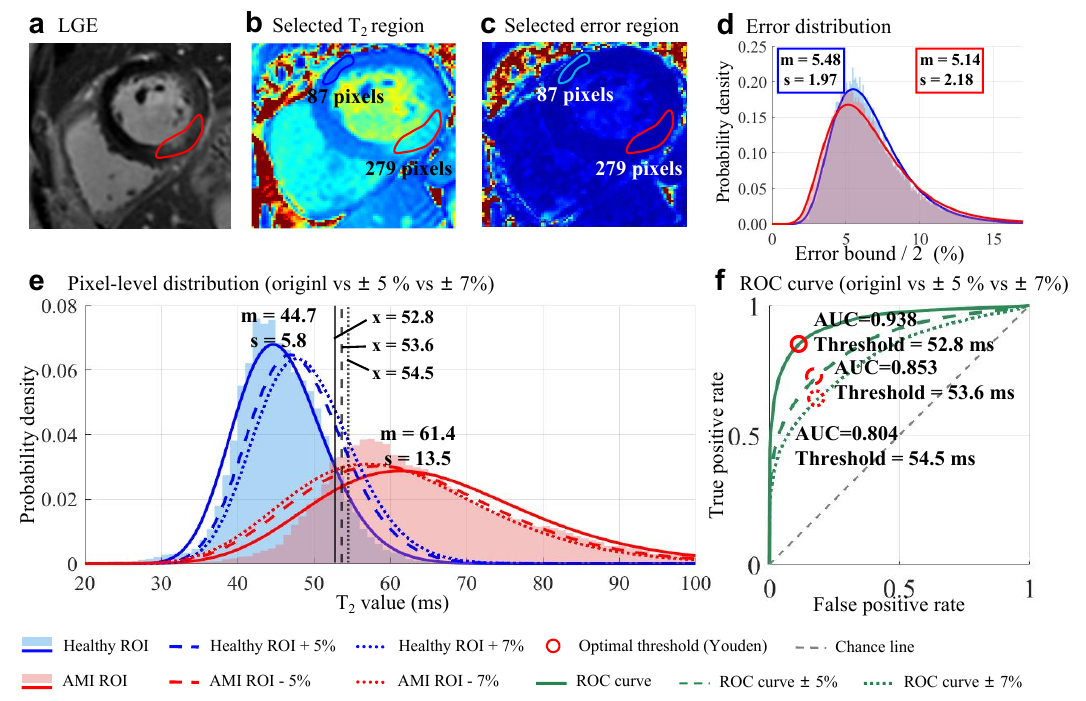}
  \caption{Pixel-wise log-normal distribution of lesion and non-lesion regions in 163 slices from 94 AMI patients. (a) Lesion region of LGE image. (b) Lesion and non-lesion regions on a $T_2$ map. (c) Lesion and non-lesion regions on the theoretical error bound/2 map. (d) Log-normal distribution of theoretical error bound/2 using mode and its standard deviation. (e) Log-normal distributions under original and 5\%, and 7\% error tolerances (original (solid), shifted 5\% (dashed), Shifted 7\% (dotted); 5\% is mode, 7\% corresponds to the mode with a standard deviation 2\%). (f) ROC curves and optimal thresholds using the Youden method under 5\% and 7\% error tolerances. Statistical values were extracted from diagonal positions of healthy regions corresponding to lesion-focused locations. $m$ and $s$ indicates the mean and standard deviation, respectively.}
  \label{figure13}
\end{figure}

\begin{figure}[]
  \centering
  \includegraphics[width=0.8\textwidth]{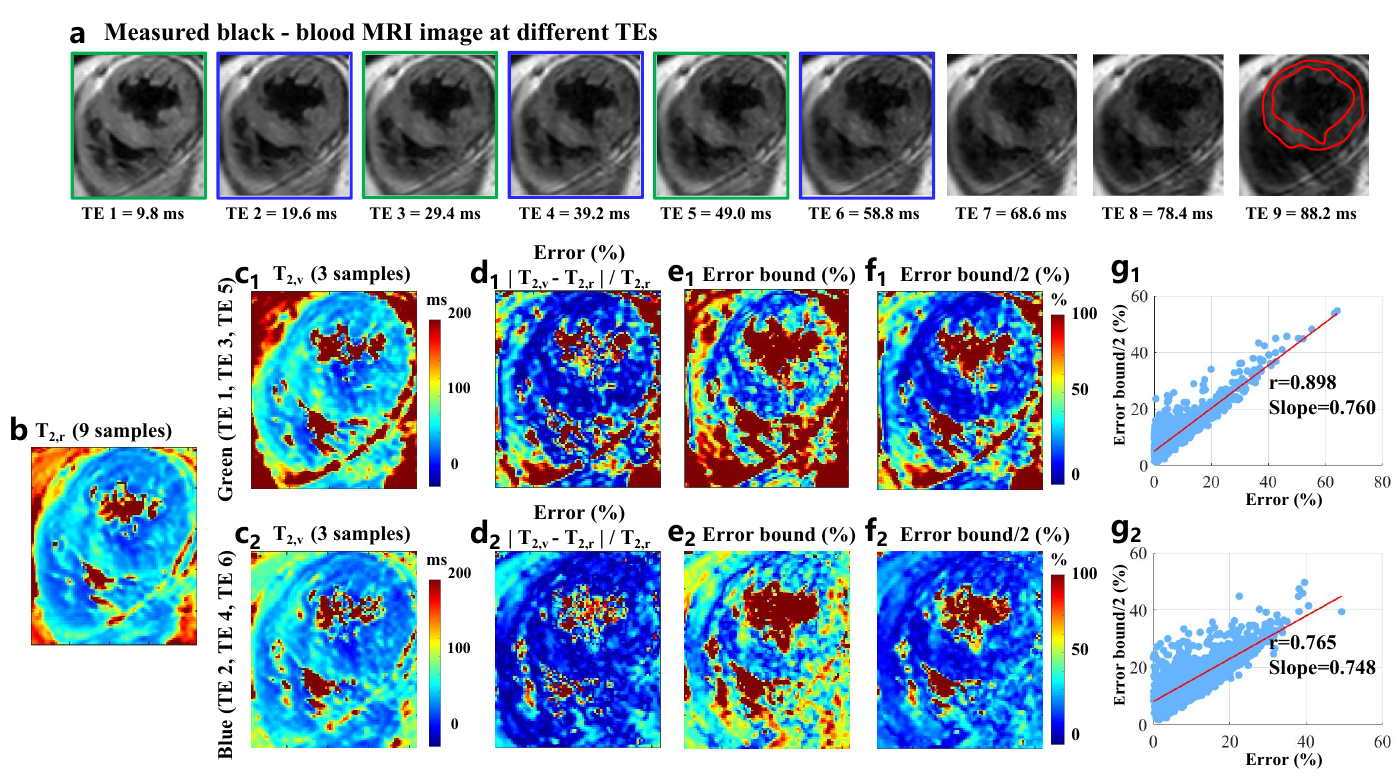}
  \caption{Validation experiment of an AMI patient. (a) is MRI images acquired at 9 TEs. (b) is the reference $T_2$ ($\mathrm{T}_{2,r}$) estimated from 9 samples. (c1) and (c2) are estimated $T_2$ ($\mathrm{T}_{2,v}$) from green-line and blue-line groups of 3 TE images in (b). (d1) and (d2) are corresponding empirical error by using (b) as reference. (e1) and (e2) are theoretical error bound. (f1) and (f2) are theoretical error bound/2. (g1) and (g2) are correlations of red-line myocardial area between the theoretical error bound/2 and empirical error (Error(\%) = $|\mathrm{T}_{2,v}$ - $\mathrm{T}_{2,r}|$/$\mathrm{T}_{2,r}$).}
  \label{figure7}
\end{figure}

  Subjective quality evaluation of $T_{2}$ map (Fig. \ref{figure10}(b)) is provided by four radiologists (20/15/11/3 years of experience) on 94 AMI patient data. In terms of clarity, localization, and interpretability, PINN achieves an average of 3.2-3.5 point (moderate to good level) on the black-blood dataset and an average of 4.0-4.3 point (good to excellent level) on the bright-blood dataset. Taking both datasets into accounts, the clarify of $T_{2}$ maps reaches good level while localization and interpretability lies in the range of moderate to good levels. Thus, the image quality of $T_{2}$ map estimated by PINN is visually applicable.

  To further validate the diagnosis value of the theoretical error bound/2, we perform the statistical analysis on the accuracy of distinguishing the lesion from non-lesion pixels. These lesion pixels are located in regions (Fig. \ref{figure13}(b)(c))) annotated by an radiologist under the guidance of LGE. A total of 163 slices of 94 AMI patients are selected by the radiologist with lesion score $>$ 2 of three metrics in Fig. \ref{figure10}(b). 

  Two log-normal distributions are observed in lesion (blue) and non-lesion (red) regions (solid curves in Fig. \ref{figure13}(e)) \cite{grinberg2012}. Normal (or abnormal) tissues in the non-lesion (or lesion) region will be misjudged being abnormal (or normal) if their $T_2$ values are larger (or lower) than a threshold, i.e. in the right (or left) side of vertical line in Fig. \ref{figure13}(e). An optimal threshold (x = 52.8 ms) is determined based on sensitivity and specificity analysis using the Youden index, achieving highly accurate classification (AUC = 0.938 in Fig. \ref{figure13}(f) (solid line)). Since modes (peak) in the log-normal distribution of error bound/2 are approximately 5\% in both regions (Fig. \ref{figure13}(d)), a worse classification is reached when the mode $T_{2}$ value of lesion (and non-lesion) tissue is increased (and decreased) by 5\%, leading to a worse threshold (x=53.6 ms). Even under this worse threshold, robust classification is still maintained (AUC=0.853, dashed line in Fig. \ref{figure13}(f)). A 7\% shift (one standard deviation) yields an AUC of 0.804 (dotted line in Fig. \ref{figure13}(f)). Both analysis demonstrate that the theoretical error bound can be used to predict the lower bound of diagnostic accuracy of lesion tissues in clinics.

Another way to validate theoretical error bound is to compare the estimated $T_2$  from few samples with an outstanding $T_{2}$ reference, which is estimated from maximal samples. Here, we compare $T_{2}$ parameters estimated from 3 samples (TE = 9.8 ms, 29.4 ms, 49 ms in green-line annotated images and TE = 19.6 ms, 39.2 ms, and 58.8 ms in blue-line annotated images in Fig. \ref{figure7}(a)) with those from 9 TE images. An empirical error is defined as a relative percentage of the estimated error over the reference $T_{2}$ ($T_{2,r}$) in the myocardium region (marked red region in the last column of Fig. \ref{figure7}(a)). The empirical error in Fig. \ref{figure7}(d1) (or (d2)) looks similar to theoretical error bound/2 in Fig. \ref{figure7}(f1) (or (f2)). Since the SNR of lower TE images are higher than larger TE ones due to exponentially declined image intensity, the former images lead to closer estimated $T_{2}$  ($T_{2,v}$) (Fig. \ref{figure7}(g1)) to the reference than the latter (Fig. \ref{figure7}(g2)). In both groups, PINN obtains strong correlations between the theoretical error bound/2 and empirical error (correlation coefficient r=0.898 in Fig. \ref{figure7}(g1) or 0.765 in Fig. \ref{figure7}(g2)).

\section{Discussion}
\subsection{More results on generalization error}
Generalization error mainly characterizes the pixel-wise difference between true pixel intensity measured in experiments and pixel intensity predicted by PINN, under different samples (TEs). For the numerical cardiac model (noise was simulated based on real cardiac data acquisition, with SNR ranging from 50.5 dB to 46.5 dB in steps of 0.5 dB, it has a real generalization error $|\mathcal{N}(t)-M_{\perp}(t)|$ over evolutionary time at each pixel. Fig. \ref{Gerror-curve} presents representative results for two pixels with $T_{2}=70$ ms: (60, 11) and (66, 12). The analysis shows that the generalization error bound can approximately capture the temporal trend of the real error. Additional results  show that fewer TEs increase both errors, while higher SNR reduces them. For an AMI patient data in Fig. \ref{AMI-verify(135)}, analysis shows that the error bound is sensitive to regions with larger deviation and this bound can roughly capture the temporal trend of the error.

\begin{figure}[]
  \centering
  \includegraphics[width=0.8\textwidth]{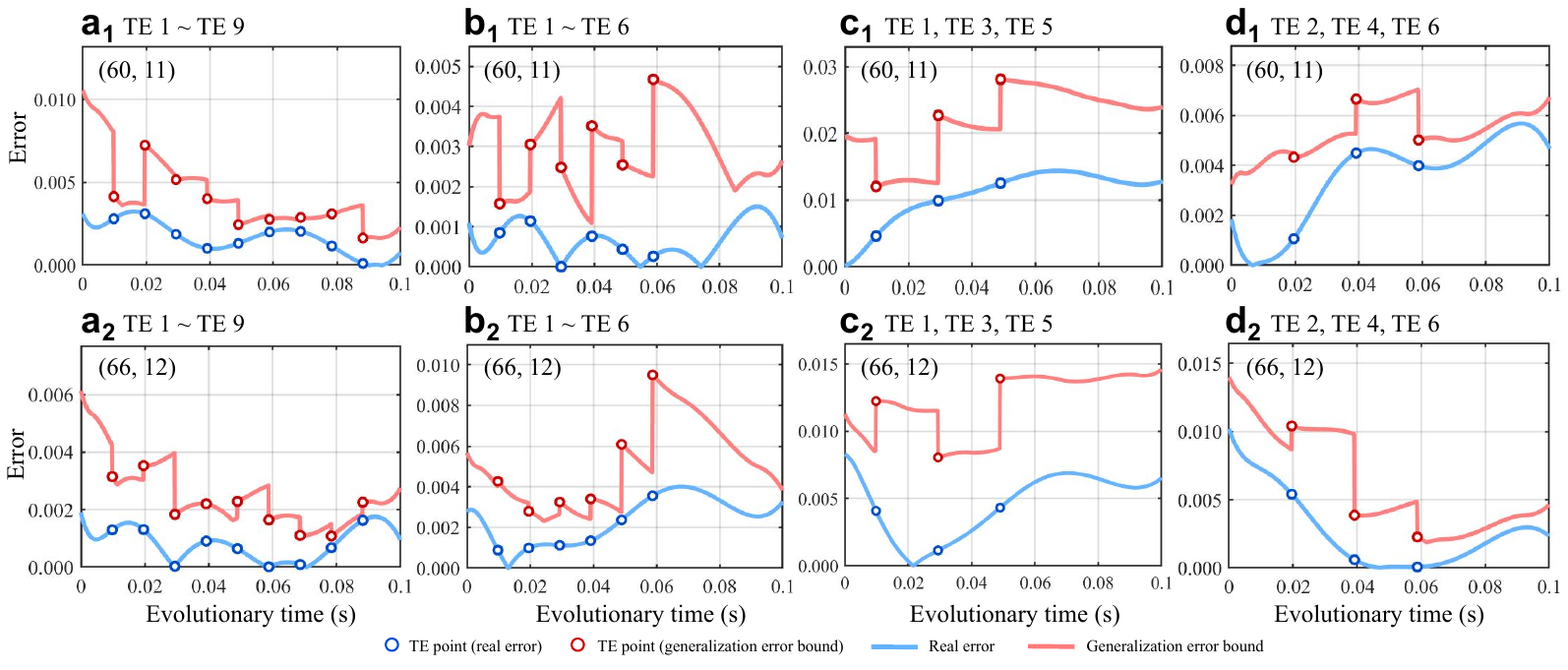}
  \caption{Generalization error  at the pixel (66, 12) and (60, 11) with $T_{2}=70 $ ms of the numerical cardiac model under different TEs. TE 1 - TE 9 were selected from 9.8 ms to 88.2 ms with an interval of 9.8 ms.}
  \label{Gerror-curve}
\end{figure}

\begin{figure}[]
  \centering
  \includegraphics[width=0.8\textwidth]{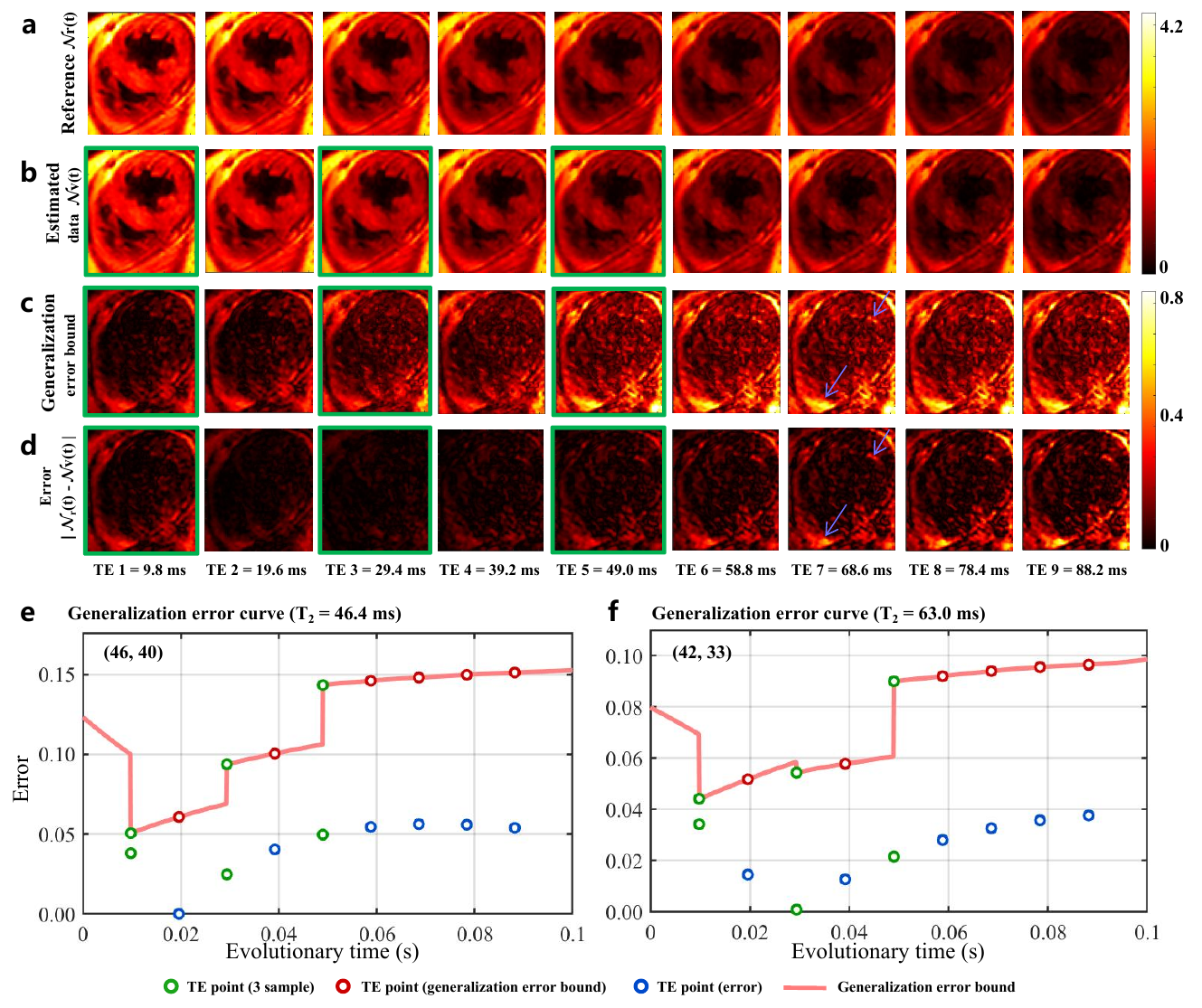}
  \caption{Validation experiment of the generalization error. (a) is the reference $\mathcal{N}_{r}(t)$ estimated from 9 samples. (b) is the validation $\mathcal{N}_{v}(t)$  $\mathrm{T}_{2,v}$ estimated of 3 samples (green frame at TE 1, TE 3, TE 5); six generated $\mathcal{N}_{v}(t)$ from $\mathrm{T}_{2,v}$ estimated of 3 samples (no green frame). (c,d) is the generalization error bound of magnetization $M_{\perp}(t)$ and error = $| \mathcal{N}_{r}(t)$ - $\mathcal{N}_{v}(t)|$. (e,f) generalization error curves on two representative pixels with $T_{2}=46.4$ ms: (46, 40) and $T_{2}=63.0$ ms: (42, 33). Blue arrows in TE = 68.6 ms shows that the error bound is very sensitive to the places where the errors are higher.}
  \label{AMI-verify(135)}
\end{figure}

\subsection{Extension of theory}
Extension to other $T_{2}$ quantification, e.g. in brain imaging, is possible since our theory was developed in the context of the general Bloch equation. This theory cloud be extended to any parametric imaging or even not imaging, that satisfies an ordinary differential equation
\begin{equation}\label{}
\frac{d U(t)}{d t} + a\cdot U(t) =0,
\end{equation}
where a single parameter $a$ is unknown. But approximating the real error with theoretical bound depends on applications. Besides, for cases involving multiple unknown parameters, e.g. multi-parametric imaging \cite{yang2023}, new theory should be established since multi-parametric differential equations and even different network architectures should be applied.

\section{Conclusion}
In this study, we proposed a physics-informed neural networks (PINN) for cardiac $T_2$ quantification and theoretically proved the upper bound of quantification error. This bound is set up in the absence of ground truth or gold-standard data, highlighting its potential clinical utility. We found that the theoretical error bound/2 is statistically close to the real quantification error in synthetic data and an industry standard phantom data. This observation is further validated in 9 TEs realistic experiments. Through the analysis of radiologist labeled lesion and non-lesion regions of myocardial infarction on two-vendor two-center two-image-type datasets, we found that the $T_2$ values in lesion (or non-lesion) satisfies log Gaussian distribution and the mode value of theoretical error bound/2 is $5.14\%$. Importantly, in phantom data, PINN outperforms both the dataset-driven deep learning method DFN and the conventional least square method, demonstrating its potential for precise $T_{2}$ quantification in real-world measurements. Further, we demonstrated the clinical value of the theoretical bound to predict the lower bound of diagnostic accuracy of lesion tissues in acute myocardial infarction. Even with no training data and ground-truth quantitative value in magnetic resonance imaging, this work provides an theoretical way to characterize the error for deep learning quantification.

\bibliographystyle{IEEEtran}

\end{document}